\documentclass{amsart}
\usepackage{amsmath}
\usepackage{amsfonts}
\usepackage{amssymb}
\usepackage{xcolor}
\usepackage{graphicx}
\newtheorem{Theorem}{Theorem}[section]

\newtheorem{Proposition}[Theorem]{Proposition}

\newtheorem{Definition}[Theorem]{Definition}
\newtheorem{rem}[Theorem]{Remark}

\newcommand{\R}{\mathbb R}
\newcommand{\K}{\mathbb{K}}

\newcommand{\N}{\mathbb N}
\newcommand{\Z}{\mathbb Z}

\newcommand{\C}{\mathbb{C}}
\newcommand{\fK}{\mathbb{K}}

\newcommand{\D}{\mathcal{D}}

\newcommand{\mbH}{\mathbb{H}}
\title[Generalized KP hierarchies and ZS-equations]{On $(t_2,t_3)-$Zakharov-Shabat equations of generalized 
Kadomtsev-Petviashvili hierarchies}

\author{Jean-Pierre Magnot$^*$}
\author{Enrique G. Reyes $^{\dagger}$}
\author{Vladimir Rubtsov $^{\star}$}

\address{$^*$  Univ. Angers, CNRS, LAREMA, SFR MATHSTIC, F-49000 Angers, France
\\ and \\  Lyc\'ee Jeanne d'Arc, \\ Avenue de Grande Bretagne, \\ 63000 Clermont-Ferrand, France}
\address{$^\dagger$ Departamento de Matem\'{a}tica y Ciencia de la Computaci\'{o}n,
	Universidad de Santiago de Chile (USACH), Casilla 307 Correo 2, Santiago,
	Chile}
\address{$^\star$ Univ. Angers, CNRS, LAREMA, SFR MATHSTIC, F-49000 Angers, France}

\email{$^*:$ magnot@math.cnrs.fr}
\email{$^\dagger:$ enrique.reyes@usach.cl}
\email{$^\star:$ volodya@univ-angers.fr}
\begin{document}

\begin{abstract}
We review the integration of the KP hierarchy in several non-standard contexts. Specifically, we consider KP in the following associative differential algebras: an algebra equipped with a nilpotent derivation; an algebra of functions equipped with a derivation that generalizes the gradient operator; an algebra of quaternion-valued functions; a differential Lie algebra; an algebra of polynomials equipped with the Pincherle differential; a Moyal algebra. In all these cases we can formulate and solve the Cauchy problem of the KP hierarchy. Also, in each of these cases we derive different Zakharov-Shabat $(t_2,t_3)$-equations ---that is, different Kadomtsev-Petviashvili equations--- and we prove existence of solutions arising from solutions to the corresponding KP hierarchy.
\end{abstract}

\maketitle


\section{Introduction}
The Kadomtsev-Petviashvili (KP) hierarchy is an integrable system in an infinite number of 
independent variables $(t_1,t_2,..)$ that contains numerous integrable equations in two and three independent variables. Among them, we mention the Boussinesq equation, the KP equation, the KdV equations and the Gelfand-Dickey flows. Usually, these equations are derived from the equations of the KP hierarchy posed on formal pseudodifferential operators $\Psi DO(A)$ with coefficients in the algebra of smooth periodic functions 
$A = C^\infty(S^1,\R)$, or in the algebra of rapidly decaying functions on the line, or in the 
Gelfand-Dickey algebraic framework explained in \cite[Chapter 1]{D}. Classical discussions of this system of equations can be found in \cite{D,MJD2000}. 
A solution to the KP hierarchy satisfies the so-called Zhakarov-Shabat (ZS) relations, as explained in \cite{D,EGR}. The ZS relations are zero-curvature equations (and hence they provide one of the justifications for stating that the KP hierarchy is an integrable system), from which the standard KP equation in the $t_2$ and $t_3$ variables can be derived.

It is already well-known that it is possible to change the algebra $A$ and define a similar KP system on $\Psi DO(A)$ even when $A$ is not commutative. We mention the paper \cite{M3} by Mulase for a very 
general construction of a KP hierarchy on a differential associative algebra, the paper \cite{EGR}, 
and also Kupershmidt treatise \cite{Ku2000} for a thorough study of the Hamiltonian content of the non-commutative 
KP hierarchy. A particular example appears in \cite{Ku2000,McI2011}, in which the algebra $A$ is the algebra 
$C^\infty(\R,\mbH),$ where $\mbH$ is the skew-field of quaternions. 
In these two references we can already see that the equations deduced from the $(t_2,t_3)-$ZS equations 
can be deeply different from the classical KP equation, see also Section 
\ref{ss:quaternion} below. Other examples of ``non-standard" KP hierarchies can be found in \cite{SZ2015,MRu2021,MR}; we also mention \cite{MR2016} for non-standard commutative examples, and  
\cite{H, Ta, Sa} as representative works on a KP hierarchy posed on algebras equipped with a Moyal multiplication.
    
Now, the construction of general solutions to the KP system for differential algebras $A$ has been systematically studied in \cite{ERMR,MR2016}: in these references, the issues of existence and  uniqueness of solutions, and well-posedness of the corresponding Cauchy problem, are formulated and stated in  frameworks as general as possible. But we remark that, strangely, the analogues of the KP equations and their solutions in generalized settings appear to
 have been mostly ignored, with very few exceptions (we mention just \cite{Ku2000,McI2011,H,MR2016}). Even some of the authors of the present work, in a first step of investigations, have been doubtful to deduce other $(t_2,t_3)-$ZS equations than those already highlighted by other authors. We show in the present work that we can deduce many types of equations: for example, we find first order nonlinear equations, Navier-Stokes-type equations, and many others. To the best of our knowledge, the equations that we present are new. 
    
    The paper is organized as follows. We recall old and not-so-old constructions for formal pseudo-differential operators and KP hierarchy on a unital differential (non-graded) associative algebra $A$ in section \ref{GKP}, and then we exhibit the announced examples in section \ref{ex}. We remark that the results of \cite{ERMR} allow us to 
    state theorems on the well-posedness our non-commutative KP hierarchies, and that these theorems imply that we can prove the existence of a full class of smooth solutions to our non-commutative $(t_2,t_3)-$ZS equations. We present these existence results also in Section  \ref{ex}. 
    
\section{The KP hierarchy of a general differential associative algebra} \label{GKP}

\subsection{Algebraic setting}

	Let $\fK$ be a characteristic zero field and let  $(A,+,* ,\partial)$ be a differential associative algebra 
	over $\fK$. 	We fix the following framework.
	
\begin{Definition}
We define  $$\Psi DO(A) = \left\{ \sum_{n \in \mathbb{Z}} a_n \xi^n : a_n \in A ; a_n = 0 \mbox{ for } n >> 0 
\right\}$$ where $\xi$ is a formal variable.
\end{Definition}

Hereafter we assume that $A$ is unital, that is, we assume that there exists $1 \in A$ such that $\forall a \in A, 1 * a = a * 1 = a.$ We call the elements of $\Psi DO(A)$ ``pseudo-differential operators" even though
we are aware that we are incurring in a slight abuse of terminology.

\begin{Definition}
Let us define addition and multiplication on $\Psi DO(A)$.
$$\begin{array}{cccc}
+ \, : &  \Psi DO(A) \times \Psi DO(A) & \longrightarrow & \Psi DO(A) \\
 & 	(\sum a_n \xi^n , \sum b_m \xi^m) & \longmapsto & \sum (a_p+b_p)\xi^p 
\end{array}$$
and 
$$\begin{array}{cccc}
	* \, : &  \Psi DO(A) \times \Psi DO(A) & \longrightarrow & \Psi DO(A) \\ 
	&  \displaystyle  	\left( a(\xi)=\sum a_n \xi^n\, , \, b(\xi)= \sum b_m \xi^m \right) & \longmapsto & 
	 \displaystyle  \sum_{m,n} \sum_{\alpha\in \N} \frac{1}{\alpha !} ( a_n  \partial^\alpha b_m) (D^\alpha_\xi\xi^n)\xi^{m} 
\end{array}$$
\end{Definition}

\begin{Theorem} 
	$(\Psi DO(A),+,*)$ is a unital (associative) $\fK -$algebra.
\end{Theorem} 
This theorem is well-known, see for example \cite{D}, or the more recent paper \cite{ERMR}. Hereafter we omit the multiplication symbol $*$ or we use simply `` $\cdot$ "\,.  
We also point out that $\Psi DO(A)$ is a {\em graded} algebra with $deg_\xi ( \sum a_n \xi^n ) = \max\{ n \in \mathbb{Z} : a_n \neq 0 \}$ and (by convention) $deg_\xi ( 0 ) = - \infty$. This grading allows us to define the vector subspaces 
$\Psi DO^k(A) = \{ P \in \Psi DO(A) : deg_\xi (P) \leq k \}$, which will appear below.

\begin{Proposition}
	The decomposition $\Psi DO(A)= DO(A)\oplus IO(A)$, in which $IO(A) = \Psi DO^{-1}(A)$ is the algebra of integral operators and $DO(A)=\Psi DO(A) \; \setminus \; \Psi DO^{-1}(A)$ is the subalgebra of differential operators, is a vector space splitting. We note that $IO(A)$ may be a 
	non-unital algebra.
\end{Proposition}
\begin{proof}
The proof relies on standard arguments:
	\begin{itemize}
		\item if $(n,m) \in \N^2,$ then for all $\alpha \in \N,$ we have that $deg_\xi \left((D^\alpha_\xi \xi^n) \xi^m\right) \in \N \cup \{-\infty\},$ which shows that $DO(A)$ is stable under multiplication, 
		\item while for $(n,m) \in (\N^*)^2,$ $\forall \alpha \in \N,$ we have that $deg_\xi \left((D^\alpha_\xi \xi^{-n}) \xi^{-m}\right) <0,$ which shows that $DO(A)$ is stable under multiplication. 
	\end{itemize}
	\end{proof}
	
\noindent	A consequence of this proposition is that any operator $a \in \Psi DO(A)$ decomposes uniquely under this direct sum as $$a = a_D + a_S $$	for $a_D \in DO(A)$ and $a_S \in IO(A)$.

\smallskip

In the most classical framework, $A$ is an algebra of functions equipped with a derivation $\partial.$  Standard choices are $A =  C^\infty(S^1,\K)$ with $\K = \R, \C$ or, in  \cite{Ku2000,McI2011}, $\K = \mathbb{H},$ and $\partial = \frac{d}{dx}.$ In these contexts, the  algebras of functions $A$ are Fr\'echet algebras equipped with natural notions of differentiability, and the operations of addition, multiplication and differentiation are smooth. It follows that addition and multiplication on $\Psi DO(A)$ are smooth in the following sense: 
if $$a = \sum_{n \in \Z} a_n \xi^n, \quad  b = \sum_{n \in \Z} b_n \xi^n ,$$
we set $$a+b = c = \sum_{n \in \Z} c_n \xi^n$$
and $$ab = d = \sum_{n \in \Z} d_n \xi^n \; .$$
Then, the map $$ \left((a_n)_{n \in \Z},(b_n)_{n \in \Z}\right) \mapsto \left((c_n)_{n \in \Z},(d_n)_{n \in \Z}\right)$$ is smooth as a map from $\Pi_\Z A \times \Pi_\Z A \rightarrow \Pi_\Z A \times \Pi_\Z A$, namely,
each coefficient $c_n$ and $d_n$ depend smoothly (in the usual way) on finite numbers of coefficients $a_n$ and $b_n$. 

The foregoing set-up circumvent the use of the technical tools recently developed in \cite{ERMR,MR2016}, where a fully rigorous framework for smoothness on objects such as $\Psi DO(A)$, is described and used. 
The geometric approach to smoothness described therein is based on the theories of diffeologies and Fr\"olicher spaces. We will not dwell on it here, we refer the reader to the mentioned papers \cite{ERMR,MR2016} for details,
but we do mention that, for instance, our point of view allows us to develop a fully rigorous approach to smoothness on Frobenius algebras ---so as to 
justify all the formal calculations of \cite{SZ2015,SZ2017,Z2017}---  
 by adapting the considerations appearing in \cite{ERMR}. 
	
	\begin{rem}
	We comment on the relation of our approach with more standard (and also more restrictive) viewpoints.
Our notion of smoothness restricts to the notion of smoothness in infinite dimensional geometry on manifolds (with 
atlases) modelled on complete locally convex topological vector spaces. This latter theory is discussed in 
\cite{Neeb}. In particular,
see \cite{MR2016},	if we work with Fr\'echet manifolds of operators, being smooth in our sense is equivalent to 
being Fr\'echet smooth. 
	\end{rem} 
	

\subsection{Integration of the KP hierarchy on a differential algebra}

Let $T=\{t_n\}_{n \in \N^*}$ be an infinite set of formal (time) variables and let us consider the 
algebra of formal series $\Psi DO(A[[T]])$ with infinite set of formal variables $t_1,t_2,\cdot$ with $T-$valuation 
$val$ defined by $val_T(t_n) = n$ \cite{M1}. 

If we assume that the algebra operations on $A$ are smooth, then so are the corresponding operations 
on $A[[T]]$, essentially because they are defined term by term using finite numbers of operations on $A$. This
intuitive description can be fully formalized in the context of ultrametric completions, see \cite{ERMR}. 
Now we extend this notion of smoothness to $\Psi DO(A[[T]])$ 
as in \cite{ERMR,MR2016}, and we are able to talk about smoothness at the level of pseudo-differential operators. We will make some further remarks on this theme after Theorem \ref{KPcentral}.

\smallskip

The Kadomtsev-Petviashvili (KP) hierarchy reads
\begin{equation} \label{eq:KP}
	\frac{d L}{d t_{k}} = \left[ (L^{k})_{D} , L \right]\; , \quad
	\quad k \geq 1 \; ,
\end{equation}
with initial condition $L(0)=L_0  \in \partial + IO(A)$.  The dependent
variable $L$ is chosen to be of the form
$$L = \xi + \sum_{\alpha \leq -1 } u_\alpha \xi^\alpha  \in {\Psi}DO^1(A[[T]]) \; .$$
A standard reference on (\ref{eq:KP}) is L.A. Dickey's treatise \cite{D}, see also 
\cite{MJD2000,M1,M3}. 

\smallskip

In order to solve the KP hierarchy, we need the following groups (see e.g. \cite{ERMR,MR2016} for the latest adaptations of Mulase's constructions appearing in \cite{M1,M3}):
$$ 
\bar{G} = 1 + IO(A[[T]])\; ,
$$

$$ \overline{\Psi} = \left\{ P = \sum_{\alpha \in {\mathbb{Z}}}
a_{\alpha}\,\xi^{\alpha}  : a_\alpha \in A[[T]]\; ,
\, val_T(a_\alpha)\geq \alpha \hbox{ and } P|_{t = 0} \in 1 + \Psi DO^{-1}(A)  \right\}\; ,$$
where $P|_{t=0}$ is the equivalence class $P \mbox{ mod } {\mathcal I}$, and ${\mathcal I}$ is the ideal of
$A[[T]]$ generated by $\{t_1,t_2, \cdots\}$,
and 
$$
\overline{\D} = \left\{ P= \sum_{\alpha \in \mathbb{Z}}
a_{\alpha}\,\xi^{\alpha} :  P \in\overline{\Psi}  \mbox{ and } a_\alpha=0 \mbox { for } \alpha <0 \right\} \; .
$$
We have a matched pair (see for example \cite{Maj}) 
$$ \overline{\Psi} = \bar{G} \bowtie \overline{\D}$$ which is smooth under the terminology we gave before.
The following result, from \cite{ERMR}, gives a synthesized statement of our results on existence and
uniqueness of solutions to the KP hierarchy (\ref{eq:KP}), and on smooth dependence on the initial conditions.

\begin{Theorem} \label{KPcentral} 
	
	~
	
	Consider the KP hierarchy \ref{eq:KP} with initial condition $L(0)=L_0$. Then,
	\begin{enumerate}
		\item There exists a pair $(S,Y) \in \bar{G} \times \overline{\D}$ 
		such that the unique solution to Equation $(\ref{eq:KP})$ with $L|_{t=0}=L_0$ is
		\begin{eqnarray*}
			L(t_1,t_2,\cdots)=Y\,L_0\,Y^{-1} = S L_0 S^{-1} \; .
		\end{eqnarray*}
		\item The pair $(S,Y)$ is uniquely determined by the smooth decomposition problem 
		$$exp\left(\sum_{k \in \N}\tau_k L_0^k\right) = S^{-1}Y\; , $$ 
		and the solution $L$ depends smoothly on the initial condition $L_0$.
		
		\item The solution operator $L$ is smoothly dependent  on the 
		initial value $L_0.$ 
	\end{enumerate}
	
\end{Theorem}

In this theorem the expression ``smoothly dependent" is to be understood essentially as explained
in the two paragraphs before Remark 2.5, namely, it means that each coefficient of the solution $L = L = \xi + \sum_{\alpha \leq -1 } u_\alpha \xi^\alpha$, understood as a series in the $T$ variables, depends smoothly on a finite number of the coefficients of the initial value $L_0.$ Indeed, if one reads the computations given in e.g. \cite{M3,ERMR,MR2016} where the computations of coefficients are sketched, it is possible to check that each coefficient of $L$ is obtained via a large but finite number of smooth operations (addition, multiplication, inversion, derivation) in $A[[T]]$ from the coefficients of $L_0.$	

\subsection{ZS-Equations deduced from the hierarchy}

It is well-known that the zero curvature equations
\begin{equation} \label{zcc0}
\frac{\partial}{\partial t_m} L^n_+ - \frac{\partial}{\partial t_n} L^m_+ = [ L^m_+ , L^n_+ ]
\end{equation}
are a consequence of the KP hierarchy, see for instance \cite[Proposition 5.1.4]{D} or \cite{EGR}, and that these 
equations are a system of non-linear equations in three independent variables, namely, $t_n,t_m$ and the ``space 
variable" implicit in the derivation $\partial$ of the differential algebra $A$. Our Theorem \ref{KPcentral} 
implies the following result (in the proposition below, we put $t_k =0$ for $k \neq n,m$):

\begin{Proposition} \label{ivp0}
Let $L = \xi + \sum_{\alpha \leq -1 } u_\alpha \xi^\alpha \in {\Psi}^1(A[[T]])$ and assume that the zero curvature
 equation $(\ref{zcc0})$ is a system of equations for $N$ dependent variables $u_1, \cdots, u_N$. Fix $u_{1,0}, 
 \cdots, u_{N,0} \in A$. Then, the system of equations $(\ref{zcc0})$ has a unique smooth solution $L$ with $u_1(0) 
 = u_{1,0}, \cdots, u_{N}(0) =u_{N,0}$, and this solution is smooth with respect to initial conditions. 
\end{Proposition}
\begin{proof}
We define $L(0)=L_0  \in \partial + IO(A)$ as 
$L(0) = \xi + u_{1,0} \xi^{-1} + \cdots + u_{N,0} \xi^{-N}$, we solve the KP hierarchy with initial condition 
$L_0$, and we apply \cite[Proposition 5.1.4]{D}.
\end{proof}

We note that this proposition does not solve the full Cauchy problem for (\ref{zcc0}), since our initial condition fixes only the ``space" dependence of the solution, while the dependence on $t_m , t_n$ is fixed by the solution
to the complete KP hierarchy with initial data $L_0$ defined as in the above proof. Nonetheless we think it is a relevant result on solutions
to very general zero curvature equations. 

\smallskip

We consider the set of independent variables 
$\{t_1,t_2,t_3\}$ and write down explicitly the corresponding zero curvature equations;  we are mostly interested
in the equations arising from the set $\{t_2,t_3\}$, but we also try the $\{t_1,t_2\}$ and $\{t_1,t_3\}$ easy cases for completeness. 

With the notation
$$ L = \sum_{ k \leq 1} u_k \xi^k$$
with $u_1 = 1$ and $u_0=0,$ we have
$$
\begin{array}{|c|c|c|}
	\hline
	& L^2 & L^3 \\ \hline
	\sigma_3 &-&\xi^3
	\\ \hline
	\sigma_2 &\xi^2&0
	\\ \hline
	\sigma_1 &0&3 u_{-1}\xi
	\\ \hline
	\sigma_0 &2u_{-1}\xi^0& (3 u_{-2} + 3 \partial u_{-1}) \xi^0 \\ \hline
\end{array}
$$
in which $\sigma_k$, $k=0,1,2,3$, indicates higher symbols (that is, homogeneous parts of degree $k$) of the 
corresponding pseudo-differential operator, and therefore
\begin{eqnarray*}
	\left[L^2_+,L^3_+\right] & = & 	\left[\xi^2,\xi^3\right] + 3	\left[\xi^2,u_{-1}\xi\right] + 3 	\left[\xi^2,u_{-2}+ \partial u_{-1}\right] \\
	&& + 2	\left[u_{-1},\xi^3\right]+ 6 	\left[u_{-1},u_{-1}\xi\right] + 6 	\left[u_{-1},u_{-2}+ \partial u_{-1}\right] \\
	& = & 0 + (6 \partial u_{-1}\xi^2 + 3 \partial^2 u_{-1}\xi) + (6 \partial u_{-2}\xi + 3 \partial^2 u_{-2} + 6 \partial^2 u_{-1}\xi + 3 \partial^3 u_{-1}) \\
	&&+ (-6\partial u_{-1}\xi^2 - 6 \partial^2u_{-1} \xi - 2 \partial^3 u_{-1}) + (6 [u_{-1},u_{-2}] - 6 \partial u_{-1} u_{-1})\\
	& = & (3 \partial^2 u_{-1}+ 6\partial u_{-2})\xi + ([u_{-1},u_{-2}]-6\partial u_{-1} u_{-1} + \partial^3 u_{-1} + 3\partial^2u_{-2}) \; .
\end{eqnarray*}

The ZS-equations for the pairs $(t_1,t_2)$ and $(t_1,t_3)$ read respectively as
\begin{equation}\label{t12}
	\frac{dL_+^2}{dt_1} - \frac{dL_+}{dt_2} = \left[L_+,L^2_+\right]
	\end{equation}
which gives: 
\begin{equation}\label{t12-deg1}
	\frac{du_{-1}}{dt_1} = \partial u_{-1}\; ,
\end{equation} 
and 
\begin{equation}\label{t13}
	\frac{dL_+^3}{dt_1} - \frac{dL_+}{dt_3} = \left[L_+,L^3_+\right]
\end{equation}
which gives:
\begin{equation}\label{t13-deg1}
	\left\{ \begin{array}{l}\frac{du_{-1}}{dt_1} = \partial u_{-1}
		\\ 
		\frac{du_{-2}}{dt_1} = \partial u_{-2} \end{array} \right.
\end{equation} 
which corresponds to the first equations of the hierarchy  $\frac{du_{-k}}{dt_1} = \frac{du_{-k}}{dx}$ when $A = (C^\infty(\R),\, \cdot\,, d/dx).$ Here, $\partial$ is an arbitrary derivation so that, instead of the usual identification $d/d{t_1} = d/dx$, we get for example the formal integration formula $$ u_{-1} = \exp(t_1 \partial) u_0$$
for an initial value $ u_{-1}|_{t_1 = 0} = u_0$.

\smallskip

We now go to the ZS-equation for the pair $(t_2,t_3),$
\begin{equation}\label{t23}
	\frac{dL^3_+}{dt_2} - \frac{dL_+^2}{dt_3} = \left[L^2_+,L^3_+\right]\; .
\end{equation}
We have:
$$
L^2_+ = \xi^2 + 2\,u_{-1} \; , \quad \quad L^3_+ = \xi^3 + 3\,u_{-1}\,\xi +3\,\partial u_{-1} + 3\,u_{-2}
$$
and
$$
[L^2_+\, ,\, L^3_+] = (3 \partial^2 u_{-1} + 6 \partial u_{-2})\xi + 6[u_{-1},u_{-2}]-6\partial u_{-1}\,u_{-1} 
+\partial^3 u_{-1} + 3 \partial^2 u_{-2}\; .
$$
Equation (\ref{t23}) yields 
 \begin{eqnarray}
 \frac{du_{-1}}{dt_2} & = & \partial^2 u_{-1} + 2 \partial u_{-2}   \label{t23-deg1prime}
		\\ 
 	\quad	\frac{d(u_{-2} + \partial u_{-1})}{dt_2} - \frac{2}{3}\frac{du_{-1}}{dt_3} & = & 2 [u_{-1},u_{-2}]-2\partial u_{-1}\, u_{-1} + \frac{1}{3} \partial^3 u_{-1} + \partial^2 u_{-2}  \; .
 																		\label{t23-deg1prime2}
 \end{eqnarray}

\smallskip

If we set $2 u_{-1} = u$, $t_2 = y$ and $t_3 =t$, system (\ref{t23-deg1prime}), (\ref{t23-deg1prime2}) becomes
\begin{eqnarray}
u_y & = & \partial^2 u + 4 \partial u_{-2}  \label{0t23-deg1} \\
\partial u_y + 2 u_{-2,y} - \frac{2}{3} u_t & = & 2[u,u_{-2}] - \partial u\,u + \frac{1}{3}\partial^3 u + 2 
\partial^2 u_{-2} \; .  \label{0t23-deg11}
\end{eqnarray}

The system (\ref{0t23-deg1}), (\ref{0t23-deg11}) reduces to the system appearing in \cite[Exercise 5.1.8]{D} if 
$u_{-1}, u_{-2}$ commute.   

We remark that this system of equations is obtained directly from our PDOs calculus without further assumptions on the derivation $\partial$. The classical KP-II equation on $C^\infty(S^1,\R)$ can be deduced from it by an adequate substitution procedure, as explained in \cite{D}. In order to work in generalized settings, we  stop our computations at system (\ref{0t23-deg1})-(\ref{0t23-deg11}), because the substitution procedure just  mentioned may not be available in them.

\section{KP equations and hierarchies in various contexts} \label{ex}
\subsection{When $\partial$ is an exterior derivation}

We say that $\partial$ is an exterior derivation if $\partial^2=0$.
In this case, 
$$u \xi^m \cdot v \xi^n = (u v) \xi^{m+n} + (m\, u\, \partial v)\, \xi^{m+n-1} \; ,$$
and therefore 
%
%
the zero-curvature equations (\ref{t23-deg1prime}) and (\ref{t23-deg1prime2}) read
\begin{equation} \label{zcc-exterior0}
\left\{ \begin{array}{rcl}
	\frac{du_{-1}}{dt_2} - 2 \partial u_{-2}&  = &0 \medskip \\
	\frac{du_{-2}}{dt_2} - \frac{2}{3}  \frac{du_{-1}}{dt_3}  & = &  - 2 \partial u_{-1} u_{-1}  + 2[u_{-1},u_{-2}]
	\; , 
\end{array}\right.
\end{equation}
in which we have used that the first equation implies that $\partial u_{-1,t_2} =0$.
%
Now we set $2 u_{-1} = u$, $t_2 = y$ and $t_3 =t$. Equation (\ref{0t23-deg1}) yields 
\begin{equation} \label{firsteq-exterior}
u = 4 \partial_y^{-1} \partial u_{-2}\; ,
\end{equation} 
and therefore equation (\ref{0t23-deg11}) becomes
$$
\partial_y u_{-2} -\frac{4}{3} \partial_y^{-1} \partial_t \partial u_{-2} = 4[ \partial_y^{-1}\partial u_{-2} ,  u_{-2} ]  \;  ,
$$
this is
\begin{equation} \label{zcc-exterior}
\partial^2_y u_{-2} -\frac{4}{3} \partial_t \partial u_{-2}  = 4\partial_y [ \partial_y^{-1}\partial u_{-2} ,  u_{-2} ] \;  ,
\end{equation}
an equation that reduces to a linear wave equation in the commutative case. Proposition \ref{ivp0} yields the result
\begin{Proposition} \label{ivp-exterior}
Let us fix $\overline{u} \in A$. Equation $(\ref{zcc-exterior})$ has a solution $u_{-2}(y,t)$ with 
$u_{-2}(0,0)= \overline{u}$, and this solution is smooth with respect to the initial condition $\overline{u}$. 
\end{Proposition}
\begin{proof}
We choose $u(y)$ such that $\partial_y u(y) = 4 \partial \overline{u}$. Such a $u(y)$ exists because the derivation 
$\partial_y$ is surjective on $A[[y,t]]$. Now we apply Proposition \ref{ivp0} with $u_{1,0} = (1/2)\,u(0)$ and 
$u_{2,0} = \overline{u}$,
thereby obtaining a solution $u_{-1}(y,t)$ and $u_{-2}(y,t)$ to (\ref{zcc-exterior0}), or, equivalently, to the 
system (\ref{firsteq-exterior})-(\ref{zcc-exterior}). By construction, the solution $u_{-2}(y,t)$ satisfies 
$u_{-2}(0,0)=\overline{u}$ and our general theory implies that it is smooth with respect to the initial condition
$\overline{u}$.
\end{proof}

\subsection{On the case when $\partial$ is the gradient} \label{s:grad}
   
Let $n \in \N^*$ and let us consider $A = C^\infty(\R^n,\R^n)$ with component-wise operations (this is an instance 
of a ``Hadamard product"). $A$ is therefore a commutative and associative algebra, but it not an integral domain.
The algebra $A$ can be equipped with the topology of uniform convergence of derivatives at any order on any compact subset of $\R^n.$ For the first part of the computations, we can also consider $A = C^\infty_c(\R^n,\R^n)$ (functions with compact support) equipped with the Whitney topology (see. e.g. \cite{KM}) or $A = S(\R^n,\R^n)$ (Schwartz space).

We denote by $x=(x_1,...x_n)^t$ a point in $\R^n$, by $u=(f_1(x),\cdots,f_n(x))^t$ a point in $A$, and we define the vector derivation $\nabla$ by 
\begin{equation} \label{grad}
\nabla u = \nabla(f_1(x),...f_n(x))^t = \left(\frac{\partial f_1}{\partial x_1}(x),...,\frac{\partial f_n}{\partial x_n}(x) \right)^t\; .
\end{equation}
Our notation is justified by the fact that ``on the diagonal", $f_1=f_2= \cdots = f_n$, our
derivation $\nabla$ is indeed standard gradient. The fact that $\nabla$ is a derivation is clear; let us
check the Leibniz rule for benefit of the reader:

\noindent We take $u=(f_1(x),\cdots,f_n(x))^t$ and $v=(g_1(x),\cdots,g_n(x))^t$ in $A$. Then (we omit the
argument $x$ for brevity),
\begin{eqnarray*}
\nabla (u \cdot v) & = & \left(\frac{\partial (f_1g_1)}{\partial x_1},...,\frac{\partial (f_n g_n)}{\partial x_n} \right)^t \\
& = & \left(\frac{\partial f_1}{\partial x_1} g_1 + f_1 \frac{\partial g_1}{\partial x_1},...,\frac{\partial f_n}{\partial x_n} g_n + f_n \frac{\partial g_n}{\partial x_n} \right)^t \\
& = & \left(\frac{\partial f_1}{\partial x_1}, \cdots , \frac{\partial f_n}{\partial x_n}\right)^t (g_1, \cdots , g_n)^t + (f_1 , \cdots f_n)^t \left(\frac{\partial g_1}{\partial x_1},...,\frac{\partial g_n}{\partial x_n} \right)^t \\
& = & \nabla(u)\, v + u\, \nabla(v)\; .
\end{eqnarray*}
\noindent We note that this computation shows that we can
replace the partial derivatives $\frac{\partial}{\partial x_1} , \cdots , \frac{\partial}{\partial x_n}$ appearing in (\ref{grad}) for any $n$-tuple of smooth vector fields $X_1, \cdots X_n$ on $\mathbb{R}^n$ and still obtain a meaningful derivation on $A$. We stick with our original
choice in view of Remark 3.2 below.

\smallskip

Equations (\ref{t23}) read as follows (for any of the three choices of $A$ mentioned at the beginning of this subsection):

\begin{equation} \label{eq:t23-gradient}
	\left\{ \begin{array}{l}\frac{du_{-1}}{dt_2} = \nabla^2 u_{-1} + 2 \nabla u_{-2}
	\medskip	\\ 
		-2\frac{du_{-1}}{dt_3} + 6u_{-1}\nabla u_{-1}   - 3 \nabla^2 u_{-1}  =   -2\nabla^3 u_{-1} -3 \nabla^2 ( u_{-2} + u_{-1}) - 3\frac{du_{-2}}{dt_2} \; . \end{array} \right.
\end{equation}

Assuming that $A= C^\infty(\mathbb{R}^n,\mathbb{R}^n)$ equipped with the smooth compact-open topology, we can 
choose $v \in C^\infty(\R^n,\R^n)[[t_2,t_3]]$ such that $\frac{d u_{-2}}{d t_2}  = \nabla (v)$.
  Setting $h_1 = 2\nabla^2 u_{-1} - 3v -3 \frac{d v}{dt_2} - 3 \nabla u_{-1}\,$, we get:
  \begin{equation}\label{t23-gradient}
  	\left\{ \begin{array}{l}\frac{d^2u_{-1}}{dt_2^2} = \nabla^2 \frac{du_{-1}}{dt_2} + 2 v
  	\medskip	\\ 
  	\nabla h_1 =	-2\frac{du_{-1}}{dt_3} + 6u_{-1}\nabla u_{-1}   -3 \nabla^2 u_{-1}  
 \medskip 	\\
  	 3 \frac{d v}{dt_2} +3v = -h_1 +2 \nabla^2 u_{-1} - 3 \nabla u_{-1}
   \end{array} \right.
  \end{equation}
where $t_2$ and $t_3$ are external $\R-$variables. We note that we obtain a ``1D Navier-Stokes-like" equation in 
the second line. In general we can make the following

\begin{rem}
We can get the Navier-Stokes equation at any dimension on $A = C^\infty(\R^n,\R^n)$ if we impose
the additional constraint $div(u_{-1})=\sum_{k=1}^n \frac{\partial u_{-1,k}}{\partial x_k}=0$. 
Indeed, let $h$ be a solution of the equation $$ \nabla h = \nabla h_1 + 3 \nabla^2 u_{-1} - \nu \Delta u_{-1}$$ we get a four lines system:
\begin{equation}\label{t23-gradient-n}
	\left\{ \begin{array}{l}\frac{d^2u_{-1}}{dt_2^2} = \nabla^2 \frac{du_{-1}}{dt_2} + 2 v
	\medskip	\\ 
	\nabla h =	-2\frac{du_{-1}}{dt_3} + 6u_{-1}\nabla u_{-1}   + \nu \Delta u_{-1} 
	\medskip	\\
		3 \frac{d v}{dt_2} +3v = -h_1 +2 \nabla^2 u_{-1} - 3 \nabla u_{-1}
	\medskip	\\
		\nabla h = \nabla h_1 +3 \nabla^2 u_{-1} - \nu \Delta u_{-1}
	\end{array} \right.
\end{equation}
The added line has a particular solution $$h = h_1 +3 \nabla u_{-1} -  \left( \nu\int_0^{x_1} \Delta u_{-1}(z_1,...z_n)dz_1, ... , \nu\int_0^{x_n}\Delta u_{-1}(z_1;...,z_n)dz_n\right)^t$$ only valid in $C^\infty(\R^n,\R^n)[[t_2,t_3]].$
\end{rem}

\subsection{On $C^\infty(\R,\mathbb{H})$} \label{ss:quaternion}
Let $\mathbb{H}$ be the field of quaternions. We give here explicitly the equations derived from the quaternionic KP hierarchy \cite{McI2011,Ku2000}, obtained for $\partial = \frac{d}{dx}$ and for $A = C^\infty(S^1,\mathbb{H})$ or $A = C^\infty(\R,\mathbb{H}). $  Equations (\ref{t23-deg1prime}) and (\ref{t23-deg1prime2}) become:
\begin{equation} \label{tt23}
	\left\{ \begin{array}{l}\frac{du_{-1}}{dt_2} = \partial^2 u_{-1} + 2 \partial u_{-2}
		\\ 
			3\frac{du_{-2} }{dt_2} - 2\frac{du_{-1}}{dt_3}  = [u_{-1},u_{-2}]-6\partial u_{-1} u_{-1} -2 \partial^3 u_{-1} - 3\partial^2u_{-2} \end{array} \right.
\end{equation}
where multiplication is the $\mathbb{H}-$multiplication, while the system (\ref{0t23-deg1}), (\ref{0t23-deg11})
become
$$
u_y = \partial^2 u + 4 u_{-2}
$$ 
and the quaternionic KP equation
$$
3 u_{yy} = \partial \left( 4 u_t - 6 \partial u\, u - \partial^3 u \right) + 3 \partial \left[ u , \partial^{-1}
u_y - \partial u \right]\; 
$$
or
\begin{equation} \label{general}
3 u_{yy} = \partial \left( 4 u_t - 3 (\partial u\, u + u\,\partial u) - \partial^3 u \right) + 3 \partial \left[ u 
, \partial^{-1} u_y  \right]\; .
\end{equation}

Let us expand $u_{-1} = a + i b + j c + k d$ and $ u_{-2} = a' + ib' + jc' + kd'$ where $(a,b,c,d,a',b',c',d') \in A^8[[T]]$. We compute: 
$$[u_{-1},u_{-2}] = \left|\begin{array}{cc}
	c & c' \\ d & d'
\end{array}\right| i + \left|\begin{array}{cc}
d' & d \\ b' & b 
\end{array}\right| j + \left|\begin{array}{cc}
b & b' \\ c & c''
\end{array}\right| k$$
and 
\begin{eqnarray*} \partial u_{-1} u_{-1} &  = & a \partial a + b \partial b + c \partial c + d \partial d \\
	&& + \left(a \partial b + b \partial a + d \partial c - c \partial d\right)i \\ 
		&& + \left(a \partial c + c \partial a + c \partial d - d \partial c\right)j \\
			&& + \left(a \partial d + d \partial a + c \partial b - b \partial c\right)k \\
			\end{eqnarray*}
		Therefore we get the following system:
		  \begin{equation}\label{t23-quaternion}
		  	\left\{ \begin{array}{l}
		  	\frac{d}{dt_2}a = \partial^2 a + 2 \partial a' \; , \quad
		  	\frac{d}{dt_2}b = \partial^2 b + 2 \partial b' \; , \quad
		  	\frac{d}{dt_2}c = \partial^2 c + 2 \partial c' \; , \quad
		  	\frac{d}{dt_2}d = \partial^2 d + 2 \partial d' \medskip 
		  		\\ 
		  		 3\frac{da'}{dt_2} -2\frac{da}{dt_3} = -6\left(a \partial a + b \partial b + c \partial c + d \partial d\right) -2 \partial^3 a - 3\partial^2 a'  \medskip
	  		\\ 
	  		3\frac{d(b' + \partial b)}{dt_2}-2\frac{db}{dt_3}  = \left|\begin{array}{cc}
	  			c & c' \\ d & d'
	  		\end{array}\right|-6\left(a \partial b + b \partial a + d \partial c - c \partial d\right) -2 \partial^3 b - 3\partial^2b'
  		\\ 
  		3\frac{dc'}{dt_2}-2\frac{dc}{dt_3}  = \left|\begin{array}{cc}
  			d' & d \\ b' & b 
  		\end{array}\right|-6\left(a \partial c + c \partial a + c \partial d - d \partial c\right) -2 \partial^3 c - 3\partial^2c'
  	\\ 
  	3\frac{d(d' + \partial d)}{dt_2}-2\frac{dd}{dt_3}  = \left|\begin{array}{cc}
  		b & b' \\ c & c''
  	\end{array}\right|-6\left(a \partial d + d \partial a + c \partial b - b \partial c\right) -2 \partial^3 u_{-1} - 3\partial^2d' \end{array} \right.
		  \end{equation}

Proposition \ref{ivp0} tells us that given $u_{-1}^0, u_{-2}^0 \in \mathbb{H}$, this system has a unique smooth solution with $u_{-1}(0) = u_{-1}^0$, $u_{-2}(0) = u_{-2}^0$, and that this solution is smooth with respect to initial conditions. We also have:

\begin{Proposition} \label{ivp-quaternion}
Let us fix $\overline{u} \in A$. Equation $(\ref{general})$ has a solution $u(y,t) \in A[[y,t]]$ with 
$u(0,0)= \overline{u}$, and this solution is smooth with respect to the initial condition $\overline{u}$. 
\end{Proposition}
\begin{proof}
We choose $u_{-2}^0=(1/4)(\overline{u} - \partial^2 \overline{u})$ and we apply Proposition \ref{ivp0} with 
$u_{1,0} = (1/2) \overline{u}$ and $u_{2,0} = u_{-2}^0$. By construction, the solution $u(y,t)=2 u_{-1}(y,t)$ 
solves (\ref{general}), it satisfies 
$u(0,0)=\overline{u}$, and our general theory implies that this solution is smooth with respect to the initial 
condition $\overline{u}$.
\end{proof}

\subsection{On the case when $\partial$ is a derivation on a Lie algebra} \label{lie}
Let $\mathfrak g$ be a Lie algebra and let $T_\mathfrak{g}$ be its tensor algebra, namely
$$
T_\mathfrak{g} = \oplus_{k \in \mathbb{N}} T^k(\mathfrak{g})\; , \quad T^k(\mathfrak{g}) = \mathfrak{g}^{\otimes k}
= \mathfrak{g} \otimes \cdots \otimes \mathfrak{g} \; , \quad (T^0(\mathfrak{g}) = \K)\; .
$$
A derivation  $D \in Der(\mathfrak g)$ induces a derivation $\partial$ on $T_\mathfrak{g}$ by the formulas: 
$$\forall \lambda \in \K, \quad \partial \lambda =0,$$
 $$ \partial(u_1 \otimes \cdots \otimes u_n) = (Du_1)\otimes u_2 \otimes \cdots \otimes u_n + u_1 \otimes (Du_2)\otimes u_3 \cdots u_n + \cdots + u_1\otimes \cdots \otimes u_{n-1}\otimes (Du_n)$$
\noindent for all $u_1, \cdots, u_n \in \mathfrak{g}$. 

Equations (\ref{t23}) read as: 
 $$\left\{ \begin{array}{l} 
 	\frac{du_{-1}^{(0)}}{dt_2} = 0 \\
 	\frac{du_{-1}^{(0)}}{dt_3} = \frac{3}{2} \frac{du_{-2}^{(0)}}{dt_3}\\
 	\frac{du_{-1}^{(1)}}{dt_2} = D^2 u_{-1}^{(1)} + 2D u_{-2}^{(1)} \\
 	 3 \frac{du_{-2}^{(1)}}{dt_2} - 2\frac{du_{-1}^{(1)}}{dt_3}  = -2D^3 u_{-1}^{(1)} -3 D^2 u_{-2}^{(1)} \\
 	\cdots
 \end{array}\right.$$

Let us specialize our computations when $D$ is an inner derivative. Let $a \in \mathfrak g.$ We set $D = [a,.]$ 
and we examine the equations obtained 
by considering the grading of the tensor algebra, 
  $$ u_{-k} = u_{-k}^{(0)} + u_{-k}^{(1)} + u_{-k}^{(2)} + \cdots$$
  Equation (\ref{t23}), in its lowest tensor orders, reads:
  $$\left\{ \begin{array}{l} 
  \frac{du_{-1}^{(0)}}{dt_2} = 0 \\
  \frac{du_{-1}^{(0)}}{dt_3} = \frac{3}{2} \frac{du_{-2}^{(0)}}{dt_3}\\
  \frac{du_{-1}^{(1)}}{dt_2} = \left[a,\left[a,u_{-1}^{(1)}\right]\right] + [2a, u_{-2}^{(1)}] \\
  3 \frac{du_{-2}^{(1)}}{dt_2} - 2\frac{du_{-1}^{(1)}}{dt_3}   = -2\left[a,\left[a,\left[a,u_{-1}^{(1)}\right]\right]\right] -3 \left[a,\left[a,u_{-2}^{(1)}\right]\right] \\
  \cdots
  \end{array}
  \right.$$  

Proposition \ref{ivp0} allow us to conclude that the Cauchy problem of this infinite system of equations can be solved in our algebraic setting. We omit the details.
  
  \smallskip
  
  Let us apply this on $\mathfrak{g} = Vect(\R^n)$ and $a = {\partial}_1 + \cdots \partial_n$, the same formulas apply and with the same notations and definitions as in section \ref{s:grad} 
   we have $[a,.] = \nabla$ as in section \ref{s:grad} but here the same equations read as: 
  $$\left\{ \begin{array}{l} 
  	\frac{du_{-1}^{(0)}}{dt_2} = 0 \\
  	\frac{du_{-1}^{(0)}}{dt_3} = \frac{3}{2} \frac{du_{-2}^{(0)}}{dt_3}\\
  	\frac{du_{-1}^{(1)}}{dt_2} = \nabla^2 u_{-1}^{(1)} + 2\nabla u_{-2}^{(1)} \\
  	3 \frac{du_{-2}^{(1)}}{dt_2}-2\frac{du_{-1}^{(1)}}{dt_3}   = -2\nabla^3 u_{-1}^{(1)} -3 \nabla^2 u_{-2}^{(1)} \\
  	\cdots
  \end{array}
  \right.$$

\subsection{On Pincherle derivative}
Let us consider $A = \K[X],$ with $\K = \R$ or $\C.$ The Pincherle derivative is defined as follows: 
$$ \forall T \in End(A), \quad  \partial_\xi T = TX - XT = [-X,T].$$
Following its fundamental properties \cite{Pin1933,RKO1973}, since $\partial_\xi X = 0$ and $\partial_\xi \frac{d}{dX}=1,$ it follows that $\partial_\xi$ is also known as a derivation on (polynomial) differential operators $A\left[\frac{d}{dX}\right].$
\begin{Proposition}
	Let $A' = \K((X))$ or $A' = \K((X^{-1})).$ i.e. 
	Then: 
	\begin{enumerate}
		\item $\partial_\xi A' = 0$
		\item $\partial_\xi$ extends to a derivation on $A'\left[\left[\frac{d}{dX}\right]\right].$
		\item $\partial_\xi$ extends to a derivation on the algebra 
		$$B = \Psi DO(A') = \left\{ \sum_{k \in \Z} a_k \xi^k : a_k \in A'\; , a_k =0, k >> 0 \right\}$$ 
		by the relation $\partial_\xi \xi = 1$.
	\end{enumerate}
\end{Proposition} 
In the sequel, the derivation in the formal $X-$variable will be noted by $\partial_X.$
\begin{proof} (1) and (2) are straightforward from the definitions. In order to prove (3), we get by induction that $\partial_\xi \xi^{-n} = -n \xi^{n+1},$ which defines it by $A'-$linearity, and then we can check that $\partial_\xi$ satisties the Leibnitz rule on $\Psi DO(A').$ 
	\end{proof}
	
We now consider $\Psi DO(B) =\left\{ \sum_{k \in \Z} a_k \zeta^k : a_k \in B, a_k =0, k >> 0 \right\}.$ 
If $P \in \Psi DO(B)$, we can write $P = \sum_{k \in \Z} a_k \zeta^k$ (or ``$P = \sum_{k \in \Z} a_k 
\partial_\xi^{\;k}\,$")
that is, $\zeta$ is the formal variable that corresponds to the derivation $\partial_\xi$ on $B$. 

Thus, we have a ``hybrid'' system, a Lie algebra bracket and a commutator in an associative algebra:
\begin{equation*}\label{t23-pinch}
	\left\{ \begin{array}{l}\frac{du_{-1}}{dt_2} = [X,[X ,u_{-1}] - 2 [ X,u_{-2}]
 \medskip		\\ 
		 	3\frac{du_{-2} }{dt_2} - 2\frac{du_{-1}}{dt_3} = [u_{-1},u_{-2}] +6[X,u_{-1}] u_{-1} +2 [X,[X,[X, u_{-1}]]] - 3[X,[X,u_{-2}]] \, . \end{array} \right. 
\end{equation*}

Once again, Proposition \ref{ivp0} allow us to conclude that the Cauchy problem of this system of equations can be solved in our algebraic setting.

\subsection{Moyal KP equation}
We work in the framework of Hamanaka's paper \cite{H}, see also \cite{Sa,Ta}. We take independent variables $x_k$ 
equipped with non-commutative multiplication
$$
[x_k , x_l] = i\, \theta^{k l}
$$
and we consider the induced multiplication on functions $f(x_1,x_2,\cdots)$ given by
$$
f(x) \star g(x) = f(x)\cdot g(x) + \frac{i}{2} \theta^{k l} \partial_{x_k} f(x) \partial_{x_l} g(x) +
O(\theta^2) \; .
$$
In this context (\ref{0t23-deg1}), (\ref{0t23-deg11}) become 
$$
u_y = \partial^2 u + 4 \partial u_{-2}
$$
and the Moyal KP equation
\begin{equation} \label{general1}
3 u_{yy} = \partial \left( 4 u_t - 3 (\partial u\star u + u\star \partial u) - \partial^3 u \right) + 3 \partial 
\left[ u , \partial^{-1} u_y  \right]\; ,
\end{equation}
where the commutator is $[f , g] = f \star g - g \star f$.
This last equation is exactly Equation (3.17) in \cite{H}. Soliton solutions to (\ref{general1}) 
have been found by Etigoff, Gelfand and Retakh, see \cite{EGR}, and also by Paniak in \cite{Pa}.
We finish this paper with a general result on the initial value problem for (\ref{general1}). We omit its proof. 

\begin{Proposition} \label{ivp-moyal}
Let us fix $\overline{u} \in (A,\star)$. Equation $(\ref{general1})$ has a solution $u(y,t) \in A[[y,t]]$ with 
$u(0,0)= \overline{u}$, and this solution is smooth with respect to the initial condition $\overline{u}$. 
\end{Proposition}

\bigskip

\paragraph{\bf Acknowledgements:} V.R. thanks IHES for hospitality; E.G.R.'s research is partially supported by the FONDECYT grant  \#1201894.

\end{document}